\documentclass[conference]{IEEEtran}
\IEEEoverridecommandlockouts
\usepackage{graphicx}
\usepackage{array}
\usepackage{amsmath}
\usepackage{float}
\usepackage{multirow}
\usepackage{array}
\usepackage{cancel}

\usepackage{tabu}
\usepackage{amsthm}
\usepackage{epsfig}
\usepackage{epstopdf}
\usepackage{epsf}
\usepackage{color}
\usepackage[english]{babel}

\theoremstyle{definition}
\newtheorem{theorem}{Theorem}
\newtheorem{lemma}[theorem]{Lemma}
\newtheorem{definition}{Definition}

\usepackage{amssymb}
\usepackage{cite}
\usepackage{amsmath,amssymb,amsfonts}
\usepackage{graphicx}
\usepackage{textcomp}
\usepackage{xcolor}
\usepackage{algorithm}
\usepackage[noend]{algpseudocode}
\usepackage[T1]{fontenc}
\usepackage{graphicx,lipsum,afterpage,subcaption}
\usepackage{enumitem}
\usepackage{lipsum}
\usepackage{mathtools}
\usepackage{cuted}

\begin{document}

	\title
	{Optimal quantizer structure for binary discrete input continuous output channels under arbitrary quantized-output constraints}

	\author{\IEEEauthorblockN{Thuan Nguyen}
		\IEEEauthorblockA{School of Electrical and\\Computer Engineering\\
			Oregon State University\\
			Corvallis, OR, 97331\\
			Email: nguyeth9@oregonstate.edu}
		\and
		\IEEEauthorblockN{Thinh Nguyen}
		\IEEEauthorblockA{School of Electrical and\\Computer Engineering\\
			Oregon State University\\
			Corvallis, 97331 \\
			Email: thinhq@eecs.oregonstate.edu}
	}

	\maketitle
	\begin{abstract}
Given a channel having binary input $X=(x_1,x_2)$ having the probability distribution $p_X=(p_{x_1},p_{x_2})$ that is corrupted by a continuous noise to produce a continuous output $y \in Y=\mathbf{R}$. For a given conditional distribution $p_{y|x_1}=\phi_1(y)$ and $p_{y|x_2}=\phi_2(y)$, one wants to quantize the continuous output $y$ back to the final discrete output $Z=(z_1,z_2,\dots,z_N)$ such that the mutual information between input and quantized-output $I(X;Z)$ is maximized while the probability of the quantized-output $p_Z=(p_{z_1},p_{z_2},\dots,p_{z_N})$ has to satisfy a certain constraint. Consider a new variable $r_y=\dfrac{p_{x_1} \phi_1(y)}{p_{x_1} \phi_1(y) + p_{x_2} \phi_2(y)}$, we show that the optimal quantizer has a structure of convex cells in the new variable $r_y$. Based on the convex cells property, a fast algorithm is proposed to find the global optimal quantizer in a polynomial time complexity. In additional, if the quantized-output is binary ($N=2$), we show a sufficient condition such that the single threshold quantizer is optimal. 
\end{abstract}

\section{Introduction}
Motivated by many applications in designing of the communication decoder i.e., polar code decoder \cite{tal2011construct} and LDPC code decoder \cite{romero2015decoding},  designing the optimal quantizer that maximizes the mutual information between input and quantized-output recently has received much attention from both information theory and communication theory society.  Over a past decade, many algorithms was proposed  \cite{kurkoski2014quantization}, \cite{zhang2016low}, \cite{iwata2014quantizer}, \cite{mathar2013threshold}, \cite{sakai2014suboptimal},  \cite{koch2013low}, \cite{kurkoski2017single}, \cite{nguyen2018capacities}, \cite{he2019dynamic}. Due to the non-linearity of quantization/partition problem, finding the global optimal quantizer is an extremely hard problem \cite{mumey2003optimal}. Therefore, most of the algorithms only can find the local optimal or  the near global optimal  quantizer \cite{zhang2016low}, \cite{mathar2013threshold}, \cite{sakai2014suboptimal},  \cite{koch2013low}, \cite{nguyen2018capacities}. However, it is well-known that if the channel input is binary, then the optimal quantizer has a structure of convex cells in the space of posterior distribution and
 the global optimal quantizer can be found efficiently in a polynomial time by using dynamic programming technique \cite{kurkoski2014quantization}. In \cite{iwata2014quantizer} and \cite{he2019dynamic}, the time complexity can be further reduced to a linear time complexity using the famous SMAWK algorithm. 

While many of works were dedicated to finding the optimal quantizer
 that maximizes the mutual information between input and quantized-output, the problem of finding the optimal quantizer under the quantized-output constrained received  much less attention. It is worth noting that finding the optimal quantizer under the quantized-output constraints having a long history. For example, the problem of entropy-constrained scalar quantization \cite{Marco2004PerformanceOL}, \cite{Gyorgy2001OnTS} and entropy-constrained vector quantization \cite{Chou1989EntropyconstrainedVQ}, \cite{Gersho1991VectorQA}, \cite{Zhao2008OnEV} were established a long time ago that aimed to minimize a specified distortion i.e., the square error distortion between the input and the quantized-output  while the entropy of the quantized-output satisfies a constraint. The constrained-entropy quantization is very important in the sense of limited communication channels. For example, one wants to quantize/compress the data to an intermediate quantized-output before transmits this quantized-output to a destination over a limited rate communication channel, then the entropy of quantized-output that denotes the lowest compression rate, is very important. Entropy-constrained can be replaced by many different output constraints i.e., power consumption constraint or time delay constraint to construct other interesting problems.  That said, the problem of quantization that maximizing mutual information under quantized-output constrained is an interesting problem and can be applied in many scenarios.
While the problem of quantization that maximizes the mutual information under quantized-output constrained is promising, there is a little of literature about this problem.  In \cite{strouse2017deterministic},  Strouse et al. proposed 
an iteration algorithm to find the local optimal quantizer that maximizing the mutual information under the entropy-constrained of quantized-output. In \cite{nguyen2019minimizing}, the authors generalized the results in \cite{strouse2017deterministic} to find the local optimal quantizer that minimizes an arbitrary impurity function while the quantized-output constraint is an arbitrary concave function. However, as the best of our knowledge, there is no work that can determine the globally optimal quantizer that maximizes the mutual information between input and quantized-output under an arbitrary quantized-output constrained even for the binary input channels. 

In this paper, we firstly show that if the channel is binary input continuous-output then for a given quantizer, there exists another \textit{convex cell quantizer} having the same quantized-output probability but produces a strictly higher or equal of mutual information between input and quantized-output. The \textit{convex cell quantizer} is a quantizer such that each quantized-output is an interval cell in space of posterior distribution. That said, to find the globally optimal quantizer, we only need to search over all the \textit{convex cell quantizers}.  Secondly, under a mild condition of quantized-output constraint, we propose a polynomial time complexity algorithm that can find the globally optimal quantizer. Finally, we characterize a sufficient condition such that a single threshold quantizer is optimal.

\section{Problem Formulation}
	\label{sec: formulation}
		\begin{figure}
		\centering
		\includegraphics[width=2.7 in]{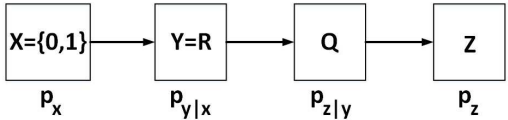}\\
		\caption{Quantization of a binary discrete input continuous output channel that maximizes the mutual information $I(X;Z)$ while the quantized-output has to satisfy a constraint $C(P_Z) \leq D$.}\label{fig: 1}
	\end{figure}		
	
Fig. \ref{fig: 1} illustrates our setting.  The discrete binary input  $X=(x_1,x_2)$  with a given pmf $p_X=\{p_{x_1},p_{x_2}\}$ is transmitted over a noisy channel. Due to the continuous noise, the output $y \in Y=\mathbf{R}$ is a continuous signal that is specified by two given conditional distributions $p_{y|x_1}=\phi_1(y)$ and $p_{y|x_2}=\phi_2(y)$. One uses a quantizer $Q$ to quantize the continuous output $y \in R$ back to the final discrete output $Z=(z_1,z_2,\dots,z_N)$ such that the the mutual information between input and quantized-output $I(X;Z)$ is maximized while the distribution of the quantized-output $p_Z=\{\ p_{z_1}, p_{z_2},\dots,p_{z_N} \}$ has to satisfy a constraint
\begin{equation}
C(p_Z)=C(p_{Z_1},p_{Z_2},\dots,p_{Z_N}) \leq D,
\end{equation}
where $C(.)$ is an arbitrary function and $D$ is a predetermined positive constant. Obviously that both $I(X;Z)$ and $p_Z$ depend on the quantizer, then we are interested in solving the following optimization problem:
	\begin{equation}
		\label{eq: jointly optimization} 
		\max_{Q} \beta I(X;Z) - C(p_Z),
	\end{equation}
where $\beta$ is pre-specified parameter to control the trade-off between maximizing $I(X;Z)$ and minimizing $C(p_Z)$.

\section{Preliminaries}
	\label{sec: review}

\subsection{Notations and definitions}

For convenience, we use the following notations and definitions:
\begin{enumerate}
	\item $r_y=p_{x_1|y}$ denotes the conditional distribution of $x_1|y$. For a given conditional distribution $\phi_1(y)=p_{y|x_1}$ and $\phi_2(y)=p_{y|x_2}$ then $r_y=\dfrac{p_1 \phi_1(y)}{p_1 \phi_1(y) + p_2 \phi_2(y)}.$
	\item $v_y=p_{x|y}=[p_{x_1|y},p_{x_2|y}]$ denotes the conditional distribution vector of $x|y$. Then $v_y=[r_y, 1-r_y]$. 
	
	\item $\mu (y)$ denotes the density distribution of variable $y$. $\mu(y)=p_1 \phi_1(y) + p_2 \phi_2(y)$. 
\end{enumerate}



\begin{definition}
\label{def: 1}
{\bf Convex cell quantizer.} 
A \textit{convex cell quantizer} is a quantizer such that each output $z_i \in Z$ is quantized by an interval cells in $r_y$ using $N+1$  thresholds $h=\{h_0=0 < h_1 <\dots, h_{N-1} < h_N=1 \}$ such that
\begin{equation}
Q(y)=z_i, \text{    if    }  h_{i-1} \leq r_y < h_i.
\end{equation}
\end{definition}

\begin{definition}
\label{def: 2}
{\bf Kullback-Leibler Divergence.} KL divergence of two probability vectors $a = (a_1, a_2, \dots, a_J)$ and $b = (b_1, b_2, \dots, b_J)$ of the same outcome set $R^J$ is defined by
\begin{equation}
\label{eq: KL divergence}
D(a||b)=\sum_{i=1}^{J}a_i \log (\dfrac{a_i}{b_i}).
\end{equation}
\end{definition}

\begin{definition}
\label{def: 3}
{\bf Centroid.} Centroid of subset $z_i$ is $ \textbf{c}_i $ which is defined by two dimensional vector $[c_i,1-c_i]$ that globally minimizes the total KL divergence $v_y$  to $\textbf{c}_i$ from all $y \in z_i$  
\begin{equation}
\textbf{c}_i=\min_{c} \int_{y \in z_i}^{} D(v_y|| c) d \mu(y). 
\end{equation}
\end{definition}

\begin{definition}
\textbf{Distortion measurement. } 
Consider a quantizer $Q$ that produces the quantized-output subsets $(z_1,z_2,\dots,z_N)$, the total distortion of $Q$ is denoted by $D(Q)$ which can be constructed by
\begin{equation}
D(Q)=\sum_{i=1}^{K} \int_{y \in z_i}^{} D(v_y||\textbf{c}_i)d \mu(y),
\end{equation}
where $\textbf{c}_i$ is the centroid of $z_i$. 
\end{definition}

\begin{definition}
\textbf{Vector order.}
Consider 2 binary vectors $v_{y_1}$ and $v_{y_2}$, we define $v_{y_1} \leq v_{y_2}$ if and only if $p_{x_1|y_1} \leq p_{x_1|y_2}$ or $r_{y_1} \leq r_{y_2}$.
\end{definition}

\begin{definition}
\textbf{Set order.}
Consider two arbitrary sets $A$ and $B$, we define $A \leq B$ if and only if for $\forall$ $y_a \in A$ and $y_b \in B$ then $v_{y_a} \leq v_{y_b}$. On the other hand, we define $A \equiv B$ if an only if $A \subset B$ and  $B \subset A$. 
\end{definition}
For example, if $z_1 \leq z_2$ then for $\forall$ $y_1 \in z_1$  and $\forall$ $y_2 \in z_2$, we have $p_{x_1|y_1}=r_{y_1} < r_{y_2} =p_{x_1|y_2}$. 

\subsection{Optimal quantizer that maximizing the mutual information is equivalent to optimal Kullback Leibler divergence distance clustering}
Interestingly, one can show that finding the optimal quantizer $Q^*$ that maximizes the mutual information $I(X;Z)$ is equivalent to determine the optimal clustering that minimizes the distortion using KL divergence as the distance metric.  The idea and proof were  already established in \cite{zhang2016low}, however, we rewrite the proof using our notation for convenient.  
For a given $y$ and a given quantizer that produces $z_i = Q(y)$ having the centroid $\textbf{c}_i$, the KL-divergence between  the conditional pmfs $v_y$  and $\textbf{c}_i$ is denoted as $D(v_y||\textbf{c}_i)$.  If the expectation is taken over $Y=\mathbf{R}$, then from Lemma 1 in \cite{zhang2016low}, we have:
\begin{eqnarray*}
\mathbb{E}_Y[D(v_y|| \textbf{c}_i )]=I(X;Y)-I(X;Z).
\end{eqnarray*}
Since $p_X$ and $\phi_i(y)$, $i=1,2$ are given, $I(X;Y)$ is given and independent of the quantizer $Q$. Thus, maximizing $I(X;Z)$ over $Q$ is equivalent to minimizing $\mathbb{E}_Y[D(v_y||\textbf{c}_i)]$ with optimal quantizer:
\begin{equation*}
\label{eq:optimal_quantizer}
Q^*=\min_{Q}  \mathbb{E}_Y[D(v_y||\textbf{c}_i)]=\min_{Q} \sum_{i=1}^{K} \int_{y \in z_i}^{} D(v_y||\textbf{c}_i)d \mu(y).
\end{equation*}

Thus, we are interested in \textit{finding the optimal quantizer $Q^*$ that minimizes the KL divergence distortion while the quantized-output satisfies a certain constraint}. 

\section{Optimal quantizer's structure}
In this section, we show that an arbitrary quantizer always can be replaced by a \textit{convex cell quantizer} with the same quantized-output while the distortion is strictly less than or equal. That said, to find the globally optimal quantizer in (\ref{eq: jointly optimization}), we only need to search over all the \textit{convex cell quantizers}. Noting that we can assume that $\textbf{c}_i \neq \textbf{c}_j$ for $i \neq j$. The reason is that if $\textbf{c}_i = \textbf{c}_j$, then one can merge $z_i$ and $z_j$ into a single subset without changing the distortion $D(Q)$. 
\subsection{Optimal structure of binary quantized-output quantizers}
We begin with the most simple scenario where the input and the quantized-output are binary. We show that for any arbitrary quantizer, existing a \textit{convex cell quantizer} having the same quantized-output distribution, however, the total distortion is strictly smaller or equal. The result is stated as follows. 
\begin{theorem}
\label{theorem: 1}
Let $Q$ is a quantizer with arbitrary two disjoint quantized-output sets $\{ z_1,z_2 \}$ corresponding to two centroids ${\textbf{c}_1,\textbf{c}_2}$ such that $\textbf{c}_1 < \textbf{c}_2$, there exists a \textit{convex cell quantizer} $\bar{Q}$ with the interval cells $\{\bar{z_1},\bar{z_2}\}$ and the corresponding centroids $\{ \bar{\textbf{c}_1},\bar{\textbf{c}_2}\}$ such that  $\{\bar{z_1} \leq \bar{z_2}\}$, $p_{z_i}=p_{\bar{z_i}}$ for $i=1,2$ and $D(\bar{Q}) \leq D(Q)$.
\end{theorem}

\begin{proof}
Due to $p_{z_1}+p_{z_2}=1$, we always can find two sets $\bar{z_1}$ and $\bar{z_2}$ such that  $\bar{z_1} \leq \bar{z_2}$ and $p_{\bar{z_i}}=p_{z_i}$ for $\forall$ $i=1,2$. 
Let $A=\bar{z_1} \cap z_2$ and $B=\bar{z_2} \cap z_1$.  Obviously that $p_{A} =  p_{B}$. 
From $\bar{z_1} \leq \bar{z_2}$, we have $A \leq B$. 
%
Now, let show that for $\textbf{c}_1=[c_1,1-c_1]<\textbf{c}_2=[c_2,1-c_2]$ then $F(r_y)=D(v_y||\textbf{c}_1)-D(v_y||\textbf{c}_2)$ is a non-decreasing function in $p_{x_1|y}=r_y$. Indeed, from the Definition \ref{def: 2},
\begin{eqnarray}
D(v_y||\textbf{c}_1)\!-\!D(v_y||\textbf{c}_2) &\!=\!& r_y  \log \dfrac{c_2(1\!-\!c_1)}{c_1(1\!-\!c_2)} + \log (\dfrac{1\!-\!c_2}{1\!-\!c_1}). \nonumber \\
\label{eq: 2-a}
\end{eqnarray}

Due to $\textbf{c}_1 < \textbf{c}_2$ implies that $c_1 < c_2$, then
$F'(r_y)=\log \dfrac{c_2(1-c_1)}{c_1(1-c_2)} > 0$. Due to the mapping from $y$ to $r(y)$ is one to one mapping (the mapping from $r(y)$ to $y$, however, may not), from $p_{A} =  p_{B}$ and $A \leq B$, then
\begin{small}
\begin{equation}
\label{eq: 2}
\int_{y \in A}^{} \! [D( v_y||\textbf{c}_1) \!-\! D( v_y||\textbf{c}_2)] d \mu(y) \! \leq \! \int_{y  \in  B}^{} \! [D( v_y||\textbf{c}_1) \!-\! D(v_y||\textbf{c}_2)] d \mu(y).
\end{equation}
\end{small}

Adding to both sides of (\ref{eq: 2}) an amount of $\int_{y\in \{z_1 \cap \bar{z_1}\}}^{} D(v_y||\textbf{c}_1)d \mu(y) + \int_{y \in \{z_2 \cap \bar{z_2}\}}^{} D(v_y||\textbf{c}_2)d \mu(y) $ and rearrange, (\ref{eq: 3}) is constructed.

\begin{strip}
\vspace{-0.1 in}
\hrule
\begin{equation}
\label{eq: 3}
\int_{y \in \bar{z_1} }^{}  D(v_y||\textbf{c}_1) d \mu(y)  +  \int_{y \in  \bar{z_2} }^{} D( v_y || \textbf{c}_2) d \mu(y)   \leq  \int_{y  \in  {z_1} }^{} D( v_y  || \textbf{c}_1 ) d \mu(y)  +  \int_{y \in  {z_2} }^{} D(v_y || \textbf{c}_2 )d \mu(y). 
\end{equation}

\begin{equation}
\label{eq: 4}
 \int_{y  \in  \bar{z_1} }^{} D  (v_y|| \bar{\textbf{c}_1})d \mu(y)  +  \int_{y \in \bar{z_2} }^{} D(v_y|| \bar{\textbf{c}_2})d \mu(y) \leq \int_{y \in  \bar{z_1} }^{} D(_y||\textbf{c}_1)d \mu(y)  +  \int_{y \in  \bar{z_2} }^{} D(v_y||\textbf{c}_2)d \mu(y).  
\end{equation}

\begin{equation}
\label{eq: 5}
 \int_{ y \in \bar{z_1} }^{} D(v_y|| \bar{\textbf{c}_1})d \mu(y)  +  \int_{y  \in  \bar{z_2} }^{} D(v_y|| \bar{\textbf{c}_2})d \mu(y) \leq  \int_{y  \in  {z_1} }^{} D(v_y||\textbf{c}_1)d \mu(y)  +  \int_{y \in {z_2} }^{} D(v_y||\textbf{c}_2)d \mu(y).
\end{equation}
\hrule
\vspace{-0.1 in}
\end{strip}

Now, by using $\bar{\textbf{c}_1}$ and $\bar{\textbf{c}_2}$ are the new centroids of $\bar{z_1}$ and $\bar{z_2}$, from Definition \ref{def: 3}, (\ref{eq: 4}) is constructed.

Finally, from (\ref{eq: 3}) and (\ref{eq: 4}), (\ref{eq: 5}) is established. That said, $D(\bar{Q}) \leq D(Q)$ which is complete our proof. 
\end{proof}
\subsection{Optimal structure of multiple quantized-output quantizers}

\begin{theorem}
\label{theorem: 2}
Let $Q$ is a quantizer with arbitrary disjoint quantized-output sets $\{ z_1,z_2,\dots,z_N \}$ corresponding to $N$ centroids ${\textbf{c}_1,\textbf{c}_2,\dots,\textbf{c}_N}$ such that $\textbf{c}_i<\textbf{c}_{i+1}$ $\forall$ $i$, there exists an other \textit{convex cell quantizer} $\bar{Q}$ with the interval cells $\{\bar{z_1},\bar{z_2},\dots,\bar{z_N}\}$ and the corresponding centroids $\{\bar{\textbf{c}_1},\bar{\textbf{c}_2},\dots,\bar{\textbf{c}_N}\}$ such that  $\bar{z_i}<\bar{z_{i+1}}$, $p_{z_i}=p_{\bar{z_i}}$ $\forall$ $i$ and $D(\bar{Q}) \leq D(Q)$.
\end{theorem}

\begin{proof}
The proof is constructed by using the induction method.  From Theorem \ref{theorem: 1}, Theorem \ref{theorem: 2} holds for $N=2$. Suppose that it also holds for $N=k$. Consider a quantizer $Q$ with arbitrary disjoint quantized-output sets $\{ z_1,z_2,\dots,z_{k+1} \}$, we show that there exists a \textit{convex cell quantizer} $\bar{Q}$ having the interval cells $\{\bar{z_1},\bar{z_2}, \dots, \bar{z_{k+1}}\}$ such that $\bar{z_i} < \bar{z_{i+1}}$, $p_{z_i}=p_{\bar{z_i}}$ $\forall$ $i=\{1,2,\dots,k+1\}$ and $D(Q) \leq D(\bar{Q})$.  
Now, without the loss of generality we suppose that
\begin{equation}
p_{z_1} =\min_{i} p_{z_i}, \text{  } \forall i.
\end{equation}
Next, we are ready to show that existing a \textit{convex cell quantizer} having the same quantized-output but the total distortion is strictly less than or equal. 
\begin{enumerate}
\item \textbf{Step 1:} Consider a \textit{convex cell quantizers} over the set $\{R / z_1 \}$. Using the assumption that the Theorem \ref{theorem: 2} holds for $N=k$, there exists a quantizer $Q^{(1)}$ which generates $\{ \bar{z_1}^{(1)},\bar{z_2}^{(1)}, \dots, \bar{z_{k+1}}^{(1)} \}$ where $\bar{z_1}^{(1)} \equiv z_1$ such that $D(Q^{(1)}) \leq D(Q)$, $p_{z_i} = p_{\bar{z_i}^{(1)}}$, $\bar{z_i}^{(1)}< \bar{z_{i+1}}^{(1)}$, $\forall$ $i \geq 2$ and $\bar{z_i}^{(1)}$  is an interval in $\{R / z_1 \}$, $\forall$ $i \geq 2$. 
 
\item \textbf{Step 2:}  Using Theorem \ref{theorem: 1} for only $\bar{z_1}^{(1)}$ and $\bar{z_2}^{(1)}$ and noting that $p_{\bar{z_1}^{(1)}}=p_{z_1} =\min_{i} (p_{z_i}) \leq p_{z_2}=p_{\bar{z_2}^{(1)}}$, existing a \textit{convex cell quantizer} $Q^{(2)}$ that generates $\{ \bar{z_1}^{(2)},\bar{z_2}^{(2)}, \dots, \bar{z_{k+1}}^{(2)} \}$ where $\bar{z_i}^{(2)} \equiv \bar{z_i}^{(1)}$, $\forall$ $i \geq 3$ and $\bar{z_1}^{(2)}$ should be the leftmost interval. That said, $\bar{z_1}^{(2)} \leq \bar{z_i}^{(2)} $, $\forall$ $i \leq 2$ and  $D(Q^{(2)}) \leq D(Q^{(1)})$. 

\item \textbf{Step 3:}  Using Theorem 2 for $N=k$ one more time over $\{ R / \bar{z_1}^{(2)} \}$, existing a convex cell quantizer $Q^{(3)}$ that generates $\{ \bar{z_1}^{(3)},\bar{z_2}^{(3)}, \dots, \bar{z_{k+1}}^{(3)} \}$ where $\bar{z_1}^{(3)} \equiv \bar{z_1}^{(2)} $ such that $D(Q^{(3)}) \leq D(Q^{(2)})$, $\bar{z_i^{(3)}}< \bar{z_{i+1}^{(3)}}$ $\forall$ $i \geq 2$. Since $\bar{z_1}^{(3)}=\bar{z_1}^{(2)}$ is the leftmost interval, $\{ \bar{z_1}^{(3)},\bar{z_2}^{(3)}, \dots, \bar{z_{k+1}}^{(3)} \}$ contains exactly $k+1$ continuous intervals such that  $\bar{z_i^{(3)}}< \bar{z_{i+1}^{(3)}}$, $\forall$ $i$. 
\end{enumerate}
Obviously that by using the \textit{convex cell quantizer} $\bar{Q}=Q^{(3)}$, the proof is complete.  
\end{proof}

\section{Discussions}
 \label{sec: algorithm}
 \subsection{Finding globally optimal quantizer using dynamic programming }
From the convex cells property of the optimal quantizer, finding the optimal quantizer is equivalent to finding $N+1$ scalar thresholds
$$h_0=0<h_1< \dots< h_{N-1}<h_N=1$$ as the boundaries such that $$Q(y) = z_i, \text{ if } h_{i-1} \leq p_{x_1|y} <h_i.$$  
Now, if the constraint of quantized-output has the following structure 
\begin{equation}
\label{eq: 100}
C(p_Z)=g_1(p_{Z_1}) + g_2(p_{Z_2})+\dots+ g(p_{Z_N}), 
\end{equation}
where $g_i(.)$ can be an arbitrary function, then the problem of finding globally optimal quantizer can be cast as a 1-dimensional scalar quantization problem that can be solved efficiently using the famous dynamic programming \cite{kurkoski2014quantization}, \cite{wang2011ckmeans}. We note that the condition in (\ref{eq: 100}) is not too restricted. In fact, many well-known constraints such
as entropy satisfy this structure.

\subsection{Binary input binary quantized-output channels and optimal single threshold quantizer}

In this section, we consider the binary input binary quantized-output channels. Due to $N=2$, from the result in Theorem \ref{theorem: 1}, the optimal quantizer can be found by searching an optimal scalar threshold $0 <a^* < 1$ such that 
$$
\begin{cases}
Q(y)=z_1 \text{   if } r_y \leq a^*,\\
Q(y)=z_2 \text{   if } r_y > a^*.
\end{cases}
$$
Thus, the optimal quantizer can be found by an exhausted searching  over a new random variable  $0 <a < 1$. The complexity of this algorithm is $O(M)$ where $M=\dfrac{1}{\epsilon}$ and $\epsilon$ is a small number denotes the precise of the solution. From the optimal value $a^*$, the corresponding thresholds $y \in Y$ can be constructed using all the solutions of $r_y=a^*$. Interestingly, the following Lemma shows a sufficient condition where a single threshold $y \in Y= \mathbf{R}$ is an optimal quantizer.  
\begin{lemma}
\label{lemma: 1}
If $\dfrac{\phi_2(y)}{\phi_1(y)}$ is a strictly increasing/decreasing function, a single threshold quantizer is optimal.
\end{lemma}

\begin{proof}
We consider
\begin{equation*}
r_y=p_{x_1|y}=\dfrac{p_1 \phi_1(y)}{p_1 \phi_1(y) + p_2 \phi_2(y)}=\dfrac{1}{1+\dfrac{\phi_2(y)}{\phi_1(y)}}.
\end{equation*}
Since $\dfrac{\phi_2(y)}{\phi_1(y)}$ is a strictly increasing/decreasing function, $r_y$ is a strictly increasing/decreasing function. Thus, for a given value of $a$, existing a single value of $y$ such that $r_y=a$. Therefore, the optimal $a^*$ corresponds to a single value of $y^*$. Thus, a single threshold quantizer is optimal in this context. Our result is an extension of Lemma 2 in \cite{kurkoski2017single}. 
\end{proof}

\section{Conclusion}
 \label{sec: conclusion}
The optimal quantizer's structure for binary discrete input continuous output channels under quantized-output constraints are explored. 
Based on the optimal structure, we proposed a polynomial time complexity algorithm that can find the globally optimal quantizer. 

\
\bibliographystyle{unsrt}
\bibliography{sample}

\end{document}